\documentclass[12pt, twoside]{article}

\usepackage{amsmath,amsthm,amssymb}

\usepackage{times}

\usepackage{enumerate}

\makeatletter

\def\author#1{\gdef\autrun{\def\and{\unskip, }#1}\gdef\@author{#1}}

\makeatother

 \newtheorem{theorem}{\bf Theorem}
 \newtheorem{lemma}[theorem]{\bf Lemma}

\newtheorem{example}[theorem]{\bf Example}

\usepackage{times}

\usepackage{enumerate}

\usepackage{amsthm}
\usepackage{amsmath}
\usepackage{amsfonts}
\usepackage{amssymb}
\usepackage{dsfont}
\usepackage{physics}
\usepackage{tikz}
\usepackage{url}
\usepackage{tkz-graph}
\usepackage{tkz-berge}
\usetikzlibrary{graphs,graphs.standard}
\usetikzlibrary{automata,arrows,positioning,calc}
\usepackage{graphicx}
\usepackage{xcolor}
\usepackage{caption} 
 \usepackage{amsmath}
\usepackage{amssymb}
\usepackage{amsthm}
\usepackage{graphicx}
\usepackage{amscd}
\usepackage{graphicx}

\usepackage{amssymb}
\usepackage{amsmath}
\usepackage{latexsym}
\usepackage[numbers,sort&compress]{natbib}
\usepackage{verbatim}
\usepackage{physics}
\usepackage{dsfont}	
\usepackage{microtype}  
\usepackage{hyperref}
\usepackage{tikz}
\usepackage{booktabs}   
\usepackage{color}

\newcommand{\one}[0]{\mathds{1}}

\def\blue\color{blue}

\usepackage{color}
\def\blue{\color{blue}}

\title{Classical and quantum cyclic redundancy check codes}

\begin{document}

\baselineskip=17pt

\author{Simeon Ball and Ricard Vilar \thanks{4 February 2025. The authors acknowledge the support of PID2020-113082GB-I00 and PID2023-147202NB-I00 financed by MICINN / AEI / 10.13039/501100011033, the Spanish Ministry of Science and Innovation.} \thanks{
Keywords: cyclic redundancy check codes, stabilizer codes, quantum cyclic redundancy check codes, syndrome decoding, linear time decoding.
MSC2020: 94B20, 94B15, 94B35. }}

\maketitle

\begin{abstract}
We prove that certain classical cyclic redundancy check codes can be used for classical error correction and not just classical error detection. We extend the idea of classical cyclic redundancy check codes to quantum cyclic redundancy check codes. This allows us to construct quantum stabilizer codes which can correct burst errors where the burst length attains the quantum Reiger bound. We then consider a certain family of quantum cyclic redundancy check codes for which we present a fast linear time decoding algorithm.
\end{abstract}

\section{Introduction} 

In many classical channels, including communication channels and storage mediums, errors tend to come in bursts, i.e. located in a short interval. There are classical codes specifically designed to deal with these types of errors. Cyclic redundancy check (CRC) codes are a class of commonly employed codes which are used to detect burst errors. These codes are cyclic codes and have the ability to detect cyclic burst errors of maximum length set out by the Reiger bound. In this article we will also show that, if chosen carefully, there are CRC codes which can also {\em correct} burst errors of maximum length set out by the Reiger bound. Analogously, in the quantum world, quantum errors can also be correlated in space and time \cite{CGLM2014} and quantum channels can introduce errors which are localised, i.e. quantum burst errors. Another objective of this article is to construct quantum CRC codes. We show that there are quantum CRC codes which meet the quantum Reiger bound. The quantum Reiger bound was proven in \cite[Theorem 3]{FHCCL2018}, where some constructions of codes meeting the bound were given. 

The {\em weight} of a vector of ${\mathbb F}_q^n$ is the number of non-zero coordinates that it has. A linear $[n,k]$ code is a $k$-dimensional subspace of ${\mathbb F}_q^n$, where ${\mathbb F}_q$ denotes the finite field with $q$ elements. We say that $C$ can correct $t$ errors if we can identify $e$, given any $u+e$, where $e$ is any vector of ${\mathbb F}_q^n$ of weight at most $t$ and $u \in C$.

\begin{theorem}[Singleton Bound \cite{Singleton1964}] If $C$ is a linear $[n,k]$ code that can correct up to $t$ errors then $n-k \geq 2t$.
\end{theorem}

We say that $C$ can correct burst errors of length $t$ if we can identify $e$, given any $u+e$, where $u \in C$ and $e$ is any vector of ${\mathbb F}_q^n$ whose non-zero coordinates lie in a burst $\{i,\ldots,i+t-1\}$, where the coordinate positions are read modulo $n$. If $t$ is minimal with this property then we say that $e$ has {\em burst weight} $t$.
Note that we could restrict to non-cyclic bursts where the coordinates are not read modulo $n$ but this would be more restrictive, so the codes we will construct will always be able to correct cyclic bursts as well. We also choose the modulo in the set $\{1,\ldots,n\}$ instead of $\{0,\ldots,n-1\}$, since we are referring to coordinates. 

The proof of the Singleton bound holds if we restrict the possible error bits to any $2t$ coordinates. Thus, the same proof gives the Reiger bound.

\begin{theorem}[Reiger Bound \cite{Reiger1960}] If $C$ is a linear $[n,k]$ code that can correct a burst error of up to length $\ell$ then $n-k \geq 2\ell$.
\end{theorem}

In this article we will also be concerned with the quantum analogues of these bounds, in particular the quantum Reiger bound. To be able to state this bound we first describe quantum errors and quantum error-correction.

\section{Quantum Error Correction} \label{qecsec}

An $(\!(n,K)\!)$ quantum error correcting code $\mathcal Q$ is a subspace of dimension $K$ in $({\mathbb C^D})^{\otimes n}$. A quantum error correcting code $\mathcal Q$ can correct arbitrary errors from a set $\mathcal E$ of unitary operators if
\begin{equation} \label{eccond}
\bra{\psi_i}E^{\dagger}E'\ket{\psi_j} = a(E,E')\delta_{ij}
\end{equation}
for all $\bra{\psi_i}\ket{\psi_j} = \delta_{ij}$ and for all $E, E' \in \mathcal E$, where $\ket{\psi_i}$ and $\ket{\psi_j} \in  \mathcal Q$, and $a(E,E')$ is a constant which depends only on $E$ and $E'$, see \cite[Theorem 10.1]{NC2000} or \cite{BrunLidar2013}. 

If condition (\ref{eccond}) holds then we can correct any linear combination of the operators in $\mathcal E$, see \cite[Theorem 10.2]{NC2000}. This observation is due to Gottesman \cite{Gottesman2009}.

We will restrict our quantum error correcting code to stabilizer codes over finite fields of prime order. Thus $\mathcal Q$ will be a subspace of $({\mathbb C}^p)^{\otimes n}$, where $p$ is a prime.

We label the coordinates of ${\mathbb C}^p$ with elements of ${\mathbb F}_p$, where ${\mathbb F}_p$ denotes the finite field with $p$ elements. In this way, a basis for the space of endomorphisms of ${\mathbb C}^p$ can be indexed by the elements of ${\mathbb F}_p \times {\mathbb F}_p$.

For each $a \in {\mathbb F}_p$, we define $X(a)$, an endomorphism of ${\mathbb C}^p$, defined as the linear map which permutes the coordinates of ${\mathbb C}^p$ by adding $a$ to the index.

In other words, the effect on the elements of the basis $\{ \ket x \ | \ x \in {\mathbb F}_p\}$ of ${\mathbb C}^p$ is
$$
X(a)\ket x=\ket{x+a}.
$$

For each $b \in {\mathbb F}_p$, we define $Z(b)$, an endomorphism of ${\mathbb C}^p$, to be the diagonal matrix whose $i$-th diagonal entry is $w^{ib}$.  Here,
$w=e^{2\pi i/p}$ is a primitive $p$-th root of unity. Thus,
$$
Z(b)\ket x=\omega^{xb}\ket x.
$$

We define the Pauli group for $p$ odd as
$$
\mathcal{P}_1=\{ \omega^cX(a)Z(b) \ | \ a,b \in {\mathbb F}_p,\ c \in {\mathbb Z}/p{\mathbb Z}\}
$$
and for $p=2$, as
$$
\mathcal{P}_1=\{ i^f \omega^c X(a)Z(b) \ | \ a,b \in {\mathbb F}_p,\ c \in {\mathbb Z}/2{\mathbb Z},\ f \in {\mathbb Z}/2{\mathbb Z}\}.
$$

More generally, we define the group of Pauli operators on $({\mathbb C}^p)^{\otimes n}$ to be
$$
\mathcal{P}_n=\{\ \sigma_1\otimes \cdots \otimes \sigma_n \ | \ \sigma_j \in \mathcal{P}_1\}.
$$
The size of $\mathcal{P}_n$ is $p^{2n+1}$ for $p$ odd and $2^{2n+2}$ for $p=2$.

The weight of an element $c\sigma_1\otimes \cdots \otimes \sigma_n $, where $\sigma_i=X(a_i)Z(b_i)$, is the number of $i \in \{1,\ldots,n\}$ such that $\sigma_i \neq X(0)Z(0)$.
Note that $X(0)Z(0)=\one$ is the identity map.

We say a code $\mathcal Q$ is a quantum $t$-error correcting code of $({\mathbb C^p})^{\otimes n}$ if condition (\ref{eccond}) holds when $\mathcal E$ is taken to be all Pauli operators of weight at most $t$.

The following bound is due to Knill and Laflamme \cite{KL1997}.

\begin{theorem}[Quantum Singleton Bound] If $Q$ is a $t$-quantum error correcting code of $({\mathbb C^p})^{\otimes n}$ and dimension $p^k$ then $n-k \geq 4t$.
\end{theorem}

The burst weight $\mathrm{bwt}(E)$ of $E=\sigma_1 \otimes \cdots \otimes \sigma_n$ is $\ell$, where $\ell$ is minimum such that all non-identity matrices $\sigma_j$ are contained in a continuous sequence of length $\ell$. In other words, there is an $r$ such that if
$\sigma_j\neq \mathrm{id}$ then  $j \in \{r,\ldots,r+\ell-1\}$. As before, the index $j$ is read modulo $n$.

The code $\mathcal Q$ is a quantum $\ell$-burst error correcting code if condition (\ref{eccond}) holds whenever $\mathcal E$ is taken to be all Pauli operators of burst weight at most $\ell$. 

The following bound is from Fan et al. \cite{FHCCL2018}.

\begin{theorem}[Quantum Reiger Bound] If there is a quantum $\ell$-burst error correcting code of $({\mathbb C^p})^{\otimes n}$ and dimension $p^k$ then
	$n-k \geq 4\ell$.
\end{theorem}

In this article, we will provide examples of families of quantum $\ell$-burst error correcting code which attain the Reiger bound and provide a linear time decoding algorithm for a particular family.

\section{Stabilizer Codes and Additive Codes}

We define the map $\tau$ from the Pauli operators to ${\mathbb F}_p^{2n}$ as 
$$
\tau(M)=(a_1,\ldots,a_n \ | \ b_1 , \ldots,b_n) \in {\mathbb F}_p^{2n}.
$$
where 
$$
M=c X(a_1)Z(b_1) \otimes X(a_2)Z(b_2)\otimes \cdots \otimes X(a_n)Z(b_n).
$$

The symplectic weight $\mathrm{swt}(v)$ of a vector $v \in {\mathbb F}_p^{2n}$ is defined as
$$
|\{ i \in \{1,\ldots,n\} \ | \ (v_i,v_{i+n}) \neq (0,0) \}|.
$$
It is immediate from the definition of $\mathrm{swt}(v)$ and $\tau$ that $\mathrm{wt}(E)=\mathrm{swt}(\tau(E))$.

The burst symplectic weight $\mathrm{bswt}(v)$ of a vector $v \in {\mathbb F}_p^{2n}$ is defined as the minimum $t$ such that there is an $r$ with the property that $(v_i,v_{i+n})= (0,0)$ for all $i \not \in \{r,\ldots,r+t-1 \}$, where indices are read modulo $n$. Again, it is clear that $\mathrm{bwt}(E)=\mathrm{bswt}(\tau(E))$.

We define an alternating form on ${\mathbb F}_p^{2n}$, where for $(a|b),(a'|b') \in {\mathbb F}_p^{2n}$,
\begin{equation} \label{symform}
((a|b),(a'|b'))_s= a\cdot b'-a'\cdot b.
\end{equation}

It is straightforward to prove that $M=X(a)Z(b)$ and $M'=X(a')Z(b')$ commute if and only if
$$
((a|b),(a'|b'))_s=(\tau(M),\tau(M'))_s=0
$$ 
and more generally that
\begin{equation} \label{commrel}
MM'=w^{(\tau(M),\tau(M'))_s} M'M.
\end{equation}

Let $S$ be a subgroup of $\mathcal{P}_n$ not containing $-\one$ and let $C=\tau(S)$. Then $C$ is an additive code of ${\mathbb F}^{2n}_p$, since $\tau(MN)=\tau(M)+\tau(N)$.

Note that $\tau$ is not a bijection since $\tau(cE)=\tau(E)$. However, we understand by $\tau^{-1}(C)$ any abelian subgroup $S$ for which $-\one \not\in S$ and $\tau(S)=C$.

The {\em centralizer} of $S$ in $\mathcal P_n$ is
$$
\mathrm{Centralizer}(S)=\{ E \in \mathcal P_n \ | \ EME^{-1}=M, \ \mathrm{for} \ \mathrm{all} \ M\in S \}.
$$
Since $S$ is abelian, the centralizer clearly contains
$$
[S]=\{ cE \ | \ E \in S, \ c^4=1\ (p=2), c^p=1,\ (p \neq 2) \}.
$$

The symplectic dual
$$
C^{\perp_s}=\{ (a',b') \in {\mathbb F}_p^{2n} \ | \ ((a|b),(a'|b'))_s=0, \mathrm{for} \ \mathrm{all} \ (a|b) \in C \},
$$
is the image of the centralizer of $S$ under the map $\tau$.  

Thus, applying the map $\tau$ we have that $S \subset \mathrm{Centraliser(S)}$ implies $C \subseteq C^{\perp_s}$.

The {\em stabilizer code} associated with $S$ is
$$
\mathcal Q(S)=\{ \ket{\psi} \in \mathcal H \ | \ M \ket{\psi}= \ket{\psi}, \ \mathrm{for} \ \mathrm{all} \ M \in S \}.
$$
All quantum error-correcting codes considered in this article will be stabilizer codes.

The following lemma is \cite[Proposition 10.5]{NC2000}.
\begin{lemma}
The dimension of $\mathcal Q(S)$ is $p^n/|S|$.
\end{lemma}

 If $\dim \mathcal Q(S)=p^k$ then we denote this as an $[\![n,k]\!]$ code.

Condition (\ref{eccond}) implies the following theorem for stabilizer codes, \cite[Theorem 10.8]{NC2000}.

\begin{theorem} If $\mathcal E$ is a subset of $\mathcal P_n$ such that  for all $E_i, E_j \in \mathcal E$ we have that $E_i^{\dagger} E_j \not \in \mathrm{Centraliser}(S)\setminus [S]$ then $\mathcal Q(S)$ can correct errors in $\mathcal E$.
\end{theorem}

If $Q(S)$ is unable to correct $E_1^{\dagger}E_2$ then $E_1^{\dagger}E_2 \in \mathrm{Centraliser}(S) \setminus [S]$. Therefore, $\tau(E_1^{\dagger}E_2)=\tau(E_1)+\tau(E_2)=v_1+v_2 \in C^{\perp_s} \setminus C$.

Thus, we have the following statement. If  for any vectors $v_1,v_2 \in {\mathbb F}_p^{2n}$ of symplectic weight  (respectively burst symplectic weight)  at most $t$, the vector $v_1+v_2 \not\in C^{\perp_s} \setminus C$ then $Q(S)$ is a quantum $t$-error correcting code (respectively quantum $t$-burst error correcting code).

We can construct a generator matrix $\mathrm{G}(S)$ for $C=\tau(S)$ by taking the $(n-k) \times 2n$ matrix 
whose $i$-th row is $\tau(M_i)$, where $M_1,\ldots,M_{n-k}$ is a set of generators for the abelian subgroup $S$.

Note that for all rows $u$ and $w$ of $G$, equation (\ref{symform}) implies
$$
\sum_{j=0}^{n} (u_{j}w_{j+n}-w_{j}u_{j+n})=0.
$$

We will make use of the following particular construction of stabilizer codes. For any matrix $H$, we denote by $H_{\ell}$ the matrix $H$ in which all columns have been shifted cyclically $\ell$ positions to the right. Thus, we have $H_0=H$ and $H_{-\ell}$ is the matrix obtained from $H$ in which all columns have been shifted cyclically $\ell$ positions to the left.

\begin{theorem}\label{quantummatrix}
Let $H$ be a $(n-k) \times n$ check matrix of a classical $[n,k]$ code over ${\mathbb F}_p$ and let $L$ be a subset of $\{1,\ldots,\lfloor (n-1)/2 \rfloor\}$. 
The $(n-k) \times 2n$ matrix 
$$
G=(H| \sum_{\j \in L} H_{+j} + H_{-j})
$$
is the generator matrix of a linear code $C=\tau(S)$, for some abelian subgroup $S$ of $\mathcal{P}_n$ of size $p^{n-k}$ defining an $[\![n,k]\!]$ stabilizer code $\mathcal Q(S)$.
\end{theorem}

\begin{proof}

By the above, it suffices to verify the condition
$$
\sum_{i=0}^{n}( u_i  v_{i+n} - v_i u_{i+n})=0,
$$
for any two rows $u$ and $v$ of $G$.

By construction, we see that $u_{i+n}=\sum_{j \in L} u_{i-j}+u_{i+j}$, where indices are read modulo $n$. Therefore,
$$
\sum_{i=0}^{n}( u_i  v_{i+n} - v_i u_{i+n})=\sum_{i=0}^{n}\sum_{j \in L} u_i  (v_{i-j}+v_{i+j}) - v_i (u_{i-j}+u_{i+j})
$$
$$
=\sum_{j \in L} \sum_{i=0}^{n} u_i  (v_{i-j}+v_{i+j}) - \sum_{j \in L}\sum_{i=0}^{n}
 (v_{i+j}u_{i}+v_{i-j}u_{i})=0.$$

\end{proof}

\section{Syndrome of a stabilizer code}
The syndrome of a (received) vector $v$ with respect to a classical linear code $C$, defined by a check matrix $H$, is 
$$
s(v)=Hv^t.
$$
 In syndrome decoding we aim to deduce the error vector $e=v-u$, where $u\in C$ is the nearest neighbour to $v$, by determining a low weight vector $e$ with the property that $s(e)=s(v)$. 
 
 In quantum error correction, we define a syndrome $s(E)  \in {\mathbb F}_p^{n-k}$ of a Pauli operator $E$ in the following way. For each generator $M_i$ of the abelian subgroup $S$, we obtain $s_i \in {\mathbb F}_p$ depending on whether
$$
M_iE\ket{\psi}=w^{s_i} EM_i \ket{\psi}=w^{s_i}E\ket{\psi}.
$$
Thus, the projective measurement with respect to the Hermitian operators $M_i$ gives us the syndrome $s(E) \in {\mathbb F}_p^{n-k}$, from which we aim to be able to identify the error $E$.

Suppose that $C=\tau(S)$ has generator matrix $G=(A|B)$, where the $i$-th row of $G$ is $\tau(M_i)$, for some set of generators $M_1,\ldots,M_{n-k}$ of $S$. 

Let $E$ be a Pauli operator and let $\tau(E)=(e_1 | e_2) \in {\mathbb F}_p^{2n}$. 

By (\ref{commrel}), if
 $$
 M_iE=w^{s_i}EM_i, 
 $$
then $s_i=(\tau(M_i),\tau(E) )_s$,
 so the syndrome of $E$ is given by
 \begin{equation} \label{syndreq}
s(E)=Ae_2^t-Be_1^t,
\end{equation}
since $s_i$ is the $i$-th coordinate of $s(E)$.

\section{Cyclic redundancy check (CRC) codes}
A cyclic redundancy check (CRC) code, defined in \cite{CRC}, is an error-detecting code commonly used in digital networks and storage devices to detect accidental changes to digital data.
In this section we will prove that certain classical CRC codes can also be used to correct errors.

For any vector $v=(v_1,\ldots,v_n)$, we associate a polynomial which we write as
$$
v(X)=\sum_{i=0}^n v_iX^{i-1}
$$
and vice-versa.

First we recall how to construct the check matrix of a classical CRC code. Let $g \in {\mathbb F}_p[X]$ be a polynomial of degree $n-k$ with $g(0) \neq 0$ and define polynomials of degree less than $n-k$ by

\begin{equation} \label{reqn}
r_i(X)=X^{i-1+n-k} \pmod{g},
\end{equation}

for $i=1,\ldots,k$. Identifying the coefficients of $r_i$ with a column vector of length $n-k$, define the $(n-k) \times n$ matrix

\begin{equation} \label{pmatrix}
\mathrm H=	\begin{pmatrix}
	\begin{array}{cccc|cccc}

		1 & 0 & \cdots & 0 & \vdots & \vdots & \cdots & \vdots\\
		0 & 1 & \cdots & 0 & r_1 & r_2 & \cdots & r_k\\
		\vdots & \vdots & \ddots & \vdots & \vdots & \vdots & \ddots & \vdots\\
		0 & 0 & \cdots & 1 & 
		\vdots & \vdots & \vdots & \vdots  \\
	\end{array}
	\end{pmatrix}.
\end{equation}

Let CRC(g) be the binary linear code with check matrix H. The following lemmas are well-known.
\begin{lemma} \label{classicalsyn} 
    The syndrome $s(e)=eH^t$ of CRC($g$) is $v$ where $v(X)=e(X) \ mod \ g$.
    
\end{lemma}

\begin{proof}
    The syndrome of the error vector $e=(e_1,\ldots,e_n)$ is $eH^t$ whose $j$-th
coordinate is
$$
(eH^t)_j=e_j+\sum_{i=1}^k r_{ij}e_{n-k+i}.
$$
which as a polynomial is
$$
\sum_{j=0}^{n-k} (eH^t)_j X^{j-1}=\sum_{j=0}^{n-k}e_jX^{j-1}+\sum_{j=0}^{n-k} \sum_{i=1}^k r_{ij}e_{n-k+i}X^{j-1}
=\sum_{j=0}^{n-k}e_jX^{j-1}+\sum_{i=1}^k e_{n-k+i} r_i(X)
$$
which is equal to $e(X)$ modulo $g$, by (\ref{reqn}).
\end{proof}

For a polynomial 
$$
f(X)=\sum_{j=1}^{n} f_j X^{j-1} \in \mathbb{F}_p[X]/(X^n-1)
$$ 
we define the (cyclic) burst weight of the polynomial $f(x)$ to be the burst weight of the vector $(f_1,..., f_n)$, as defined before.

\begin{lemma} \label{non-zero}
    If $g$ divides $X^n-1$ then a non-zero multiple of $g$ in $\mathbb{F}_p[X]/(X^n-1)$ has burst length at least $\deg(g)+1$.
\end{lemma}

\begin{proof}
There is no multiple of $g$ which is of smaller degree in the ring ${\mathbb F}_p[X]/(X^n-1)$, since if $$a(X)g(X)+b(X)(X^n-1)=c(X)$$ then $g$ divides $c(X)$. 

Suppose $e(X)$ is a multiple of $g$ whose burst weight is less than the $\deg(g)+1$. Then there exists $m$ such that $X^m e(X)$ has degree less than degree of $g$, and is a multiple of $g$, a contradiction.

\end{proof}

Let $g$ be a polynomial of degree $n-k$ dividing $X^n-1$. 

The following theorem is from \cite{CRC}.

\begin{theorem}
    The code CRC($g$) detects errors of burst weight at most $n-k$. 
    \end{theorem}

\begin{proof}
Let $v=u+e$, where $u\in C$ and $e$ is the error vector of burst weight at most $n-k$. 

Note that if the syndrome $s(v)=s(e) \neq 0$ then an error has occurred and is detected. 

If $s(e)=0$ then by Lemma~\ref{classicalsyn}, the polynomial $e(X)$, identified with the error vector $e$, is a non-zero polynomial of degree at most $n-k-1$ which is zero modulo $g$. Thus, $e(X)$ is a multiple of $g$. By Lemma \ref{non-zero}, since $e$ has burst weight at most $n-k$, we conclude $e=0$. 

\end{proof}
 
 We now introduce a new property for polynomials called the c-property. We say that $g$ has the {\em c-property} if no multiple of $g$ in $\mathbb{F}_p[X]/(X^n-1)$ is the sum of two polynomials of (cyclic) burst weight at most  $\lfloor \frac{n-k}{2} \rfloor$. Observe that this is equivalent to choosing $g$ so that the $k$-dimensional cyclic code generated by $g$ contains no codewords which are the sum of two vectors of burst weight at most $\lfloor \frac{n-k}{2} \rfloor$. 
 
 The following theorem appears not to have been observed before.
 
 \begin{theorem} \label{lemcorr}
 If $g$ has the c-property then CRC($g$) corrects cyclic burst errors of length at most $\lfloor (n-k)/2 \rfloor$.
\end{theorem}

\begin{proof}
Let $e_1$ and $e_2$ be errors of burst weight at most $\lfloor \frac{n-k}{2} \rfloor$.
If $s(e_1)=s(e_2)$ then $s(e_1-e_2)=0$. By Lemma \ref{classicalsyn}, $(e_1-e_2)(X)$ is a multiple of $g$, which contradicts the c-property. Thus, $s(e_1)\neq s(e_2)$.
 
 Thus, if $v=u+e$, where $u\in C$ and $e$ is the error vector of burst weight at most $\lfloor \frac{n-k}{2} \rfloor$ then $s(e)=s(v)$ and since $e$ can be identified by its syndrome, we can determine $u$.
\end{proof}

Let $g_i$ be the coefficient of $X^i$ in $g(X)$, and define a $k \times n$ generator matrix for the cyclic code $\langle g \rangle$ of length $n$, 
$$
\mathrm{G}=\left(
\begin{array}{ccccccc}
 g_0 & \ldots & g_{n-k} & 0 & \ldots & \ldots & 0 \\
0 & g_0  & \ldots & g_{n-k} & 0 & \ldots & 0 \\
0& 0 & \ddots & \ldots & \ddots & \ddots & \vdots \\
\vdots & & \ddots & \ddots & & \ddots & 0 \\
0 & \ldots & \ldots & 0 & g_0 & \ldots & g_{n-k}\\
\end{array} \right).
$$
If $n-k$ is odd then let $\mathrm{G}_r$ be the $(r+1)\times (r+2)$ matrix which is the submatrix of $\mathrm{G}$ formed from the $((n-k+1)/2)$-th to $((n-k+1)/2+r+1)$-th column and from the first to the $(r+1)$-th row.

If $n-k$ is even then let $\mathrm{G}_r$ be the $(r+1)\times (r+1)$ matrix which is the submatrix of $\mathrm{G}$ formed from the $((n-k)/2+1)$-th to $((n-k)/2+r+1)$-th column and from the first to the $(r+1)$-th row.

\begin{theorem} \label{cpropcheck}
The polynomial $g$ has the c-property if and only if for all $r \in \{0,\ldots,k-1\}$ the matrix $\mathrm{G}_r$ has rank $r+1$.
\end{theorem}

\begin{proof}
The polynomial $g$ does not have the c-property if and only if the matrix $G$ has a non-zero vector in its row span whose support is contained in two bursts of length at most $\lfloor (n-k)/2 \rfloor$.
Since $\langle g \rangle$ is a cyclic code this is if and only if the matrix $G$ has a non-zero vector $u$ in its row span whose support is contained in two bursts of length at most $\lfloor (n-k)/2 \rfloor$, whose first coordinate is non-zero. 

Let $r$ be maximum such that $u_{n-k+r+1} \neq 0$. Then the support of $u$ is contained in the first $\lfloor (n-k)/2 \rfloor$ coordinates and the $(n-k+r+2-\lfloor (n-k)/2 \rfloor)$-th to the $(n-k+r+1)$-th coordinate. Furthermore, $u$ is a linear combination of the first $r+1$ rows of $G$, since in the $i$-th row of $G$, the $(n-k+i)$-th coordinate is non-zero. Thus, $u$ is a linear combination of the first $r+1$ rows of $G$ in which the $\lfloor (n-k)/2 \rfloor+1$ to $(n-k+r+1-\lfloor(n-k)/2 \rfloor)$-th coordiantes are zero. This is precisely the condition that  $\mathrm{G}_r$ has rank at most $r$.
\end{proof}

Theorem~\ref{cpropcheck} gives a fast algorithm for checking if $g$ has the c-property. In fact, since $\langle g \rangle$ is cyclic one can assume that the string of $r+1$ zeros in $u$ is no longer than the string of $k-r-1$ zeros from the $(n-k+r+2)$-th coordinate to the $n$-th coordinate. Thus it is sufficient to check the rank of $G_r$ for $r\in \{0,\ldots,(k-2)/2\}$.

Tables~\ref{cproptable1}, \ref{cproptable2},
list all polynomials $g \in {\mathbb F}_2[X]$ which are factors of $X^n-1$ and have the c-property for all $n \leqslant 27$. We have omitted $g(X)=(X^n-1)/(X-1)$ and $g(X)=X-1$ from the table since these always have the c-property.

In \cite{AMD2021}, some analysis was carried out showing that decoding algorithms such as GRAND, can be used to correct errors using CRC codes. These results do not make use of polynomials which necessarily have the c-property. It would be interesting to have more sophisticated decoding algorithms which were tailored to a particular choice of $g$ and make use of Theorem~\ref{lemcorr}.

\section{Quantum CRC codes}

In this section we construct a quantum CRC code from a classical CRC code and analyse its error correction and detection properties. 

For the remainder of the article, we restrict to the case $p=2$.

Let $H$ be defined as in (\ref{pmatrix}) for this $g$.

By Lemma \ref{quantummatrix}, the matrix

\begin{equation} \label{Gmatr}
G=(H|H_{+\ell}+H_{-\ell})
\end{equation}

is the generator matrix of a linear code $C(S)$, for some abelian subgroup $S$ of $\mathcal{P}_n$ defining $\mathcal Q_g(S)$, an $[\![n,k]\!]$ quantum stabilizer code, which depends only on the chosen polynomial $g$. 

We fix $\ell=\lfloor (n-k)/4 \rfloor$.

For a Pauli error $E \in \mathcal{P}_n$, we define 
$$
e=(e_1 | e_2)=\tau(E) \in {\mathbb F}_2^{2n}
$$ 
which we identify with a polynomial of degree at most $2n-1$,

$$e(X)=e_1(X) + X^n  e_2 (X),$$

by identifying $e_1, e_2 \in  {\mathbb F}_2^{n}$ with a polynomial of degree at most $n-1$ as before.

As we saw in (\ref{syndreq}), the syndrome of an error $e=(e_1 | e_2)$, with respect to an $(n-k)\times 2n$ matrix $(A | B)$ is
$$
s(e)=Ae_2^t -Be_1^t.
$$

\begin{lemma} \label{synlem}
    The syndrome of an error $e(X)=e_1(X) + X^n e_2(X)$, with respect to the code $\mathcal Q=\mathcal Q_g(S)$, is
    $$s_\mathcal Q (e)= X^{-\ell}e_1+X^{\ell} e_1 + e_2 \pmod{g}.$$
\end{lemma}

\begin{proof}
With respect to the code $\mathcal Q=\mathcal Q_g(S)$, the syndrome is calculated using the matrix $G=(H|H_{+\ell}+H_{-\ell})$. By Lemma~\ref{classicalsyn}, the $e_2$ part of the error contributes $e_2(X)$ modulo $g$ to the syndrome and the $e_1$ part of the error contributes $X^{-\ell}e_1+X^{\ell} e_1$ modulo $g$ to the syndrome.

\end{proof}

\begin{theorem} \label{thm213}
 The $[\![n,k]\!]$ quantum CRC code $\mathcal Q_g(S)$ detects errors of burst weight at most $2\ell$.\end{theorem}

\begin{proof}
Suppose that $e(X)=e_1(X) + X^n e_2(X)$ is a polynomial of symplectic burst weight at most $2\ell$.

We will prove that if $s_Q(e)=0$ then $e=0$. 

Suppose that $s_Q(e)=0$ and that the burst starts in the $r$-th coordinate. Then 
$$
e_1(X)=X^{r-1}f_1(X),\ \mathrm{and} \ e_2(X)=X^{r-1}f_2(X)
$$
for some polynomials $f_i(X)$ of degree at most $2\ell-1$.

By Lemma~\ref{synlem},
$$ X^{-\ell}e_1+X^{\ell} e_1 + e_2=0 \pmod{g}$$
which implies
$$ X^{-\ell}f_1+X^{\ell} f_1 + f_2=0 \pmod{g}.$$

By Lemma \ref{non-zero}, either 
$$
X^{-\ell}f_1+X^{\ell} f_1 + f_2=0  \pmod{X^n-1}
$$ 
or the burst length of 
$X^{-\ell}f_1+X^{\ell} f_1 + f_2$ is at least $n-k+1$. But  the burst length of $X^{-\ell}f_1+X^{\ell} f_1 + f_2$ is at most $2\ell+\deg(f_1)+1 \leqslant 4\ell$. Hence,
$$
X^{-\ell}f_1+X^{\ell} f_1 + f_2=0 \pmod{X^n-1}.
$$
Thus, for $j=0,\ldots,\ell-1$, the coefficients of $X^j$ in $f_1$ are zero. Since the degree of $f_2$ is at most $2\ell-1$, it follows that, for $j=0,\ldots,\ell-1$, the coefficients of $X^{\ell+j}$ in $f_1$ are also zero. Hence, $f_1=0$, from which it follows that $f_2=0$ and thus, $e=0$.
\end{proof}

\begin{theorem}
 If $g$ has the c-property then the $[\![n,k]\!]$ code $Q_g(S)$ can correct errors of burst weight at most $\ell$.
\end{theorem}

\begin{proof}
Suppose that $e(X)=e_1(X) + X^n e_2(X)$ and $e'(X)=e'_1(X) + X^n e'_2(X)$ are polynomials of symplectic burst weight at most $\ell$.

We will prove that if $s_Q(e)=s_Q(e')$ then $e=e'$. 

By Lemma \ref{synlem},
$$ X^{-\ell}e_1+X^{\ell} e_1 + e_2-( X^{-\ell}e_1'+X^{\ell} e_1' + e_2')=0 \pmod{g}.$$

As in the proof of Theorem~\ref{thm213}, $X^{-\ell}e_1+X^{\ell} e_1+e_2$ has burst weight at most $2 \ell$.



This implies there is a multiple of $g$ in the ring ${\mathbb F}_2[X]/(X^n-1)$ which is the sum of 
two polynomials of length at most  $2\ell$. Since $g$ has the c-property this implies 
$$ X^{-\ell}e_1+X^{\ell} e_1 + e_2-(X^{-\ell}e_1'+X^{\ell} e_1' + e_2')=0 \pmod{X^n-1}.$$
from which we deduce, as in the previous proof, that $e=e'$.
\end{proof}

In the next section, we will give an example of $g(X)$ for which the quantum CRC code constructed as in this section can be equipped with a linear time decoding algorithm capable of correcting all errors of cyclic burst length at most $\lfloor (n-k)/4\rfloor$.

\section{A fast decoding algorithm for a particular family of quantum CRC codes}

In this section we will detail a fast decoding algorithm for a particular family of qubit quantum CRC codes. We suppose $n=mk$ and $l=ck$ where $m$ and $c$ are whole numbers, taken to be constants. Observe that the quantum Reiger bound implies $m \geqslant 4c+1$, which implies that the rate of the code is at most $1/(4c+1)$.

Note that in the qubit case, we denote the Pauli matrices by
$$
I=\begin{pmatrix}1&0\\0&1\end{pmatrix}, X=\begin{pmatrix}0&1\\1&0\end{pmatrix},
Y=\begin{pmatrix}0&-i\\i&0\end{pmatrix},  
Z=\begin{pmatrix}1&0\\0&-1\end{pmatrix}
$$
and so, as defined in Section~\ref{qecsec},
$$
\mathcal P_n:=\{ c\sigma_1 \otimes \cdots \otimes \sigma_n \ | \ \sigma_j \in \{I,X,Y,Z\} \}.
$$

\begin{theorem}
If $g(X)=X^{n-k} + X^{n-2k}+ \dots +X^k +1$ then there is a decoding algorithm for $Q_g(S)$ whose complexity scales linearly with n.   
\end{theorem}

\begin{proof}

The choice of $g$ implies that the matrix $G$, defined as in (\ref{Gmatr}) will have a quasicyclic structure. 

The subgroup $S$ which stabilises the $[\![n,k]\!]$ code $\mathcal Q(S)$ can be constructed by interleaving with $k$ base $[\![m,1]\!]$ codes, which are constructed from a subgroup $S_j$ of $({\mathbb C}^2)^{\otimes m}$, where $S_j$ is obtained from $S$ considering only those qubits (coordinates) which are congruent to $j$ modulo $k$. Hence, $j \in \{0,\ldots,k-1\}$.

Suppose that we have measured the quantum system and obtained a syndrome $s \in \{+,-\}^{n-k}$.

We define the subsyndrome $s_j$ to be the restriction of $s$ to the $i$ coordinates where $i=j$ modulo $k$, so $s_j \in \{+,-\}^{m-1}$. Note that here we use $(+,-)$ in place of $(0,1)$ which we used for the classical syndromes.

We will determine the error $E \in \mathcal P_n$ by determining $E_j \in \mathcal P_m$, for each $j \in \{0,\ldots,k-1\}$.

 If we undo the interleaving of the matrix, we will get $k$ identical matrices of the dimension "1" code. We can see, by the definition of the subsyndrome and the interleaving of the matrix, that each subsyndrome will only depend on the dimension 1 matrix and the qubits (coordinates) of the quantum system which are congruent to $j$ modulo $k$.

Firstly, we will analyze the construction of a dimension $k=1$ code matrix since it is this matrix alone that determines the subsyndrome.

Considering the construction of the matrix (\ref{Gmatr}), we see that the dimension $k=1$ code matrix has on the left side an $(m-1) \times (m-1)$ identity matrix where the $m$-th column is the all one column. On the right side of the matrix, we will have the sum of this matrix shifted $c$ positions to the right and to the left.

If we look at this matrix, see for example (\ref{eqn11}), we can clearly see 3 columns whose weight is more than $2$. The last column of the left-hand side matrix is the all $1$ vector, and the $c$ and $m-c$ column of the right-hand side are the all $1$ vector summed with a vector of weight one. 

The syndromes which are affected by these columns are those which have Pauli errors in the positions indicated in Table~\ref{dodgyerrors}.

\begin{table}[h]
    \centering
    
    \begin{tabular}{|c|c|}
    \hline
    1 & An $X$ error in the $c$ position.\\
  2 &  An $X$ error in the $m-c$ position.\\
   3 & A $Z$ error in the $m$ position.\\
   4 & A $Y$ error in the $m$ position.\\
   5 & A $Y$ error in the $c$ position.\\
   6 &  A $Y$ error in the $m-c$ position.\\
   \hline
    \end{tabular}
    \caption{The errors which are located by the sub-syndrome look-up table.}
    \label{dodgyerrors}
\end{table}

We construct a look up table for all possible errors that fit one of these errors 1-6.
Since, by assumption, the error $E$ has a burst length of length at most $\ell$, it follows that $E_j$ has a maximum burst length of $c$. Thus, one can only have at most one of the above errors 1-6 occurring. Hence, the number of syndromes in the look up table is bounded by the number of possible combinations with repetition of $X,Y,Z,I$ of length $c-1$, $4^{c-1}$ multiplied by the number of different possible errors, $6$, multiplied by $c$, since the sub-burst error, with an error from Table~\ref{dodgyerrors}, can start in any one of $c$ positions.

Hence, we will have at most $4^{c-1}6c$ syndromes in the look-up table. 

Since we consider $c$ as a constant, this look-up table will have just a constant number of entries.

 The different syndromes for the individual errors cannot collide. The $Z$ error clearly can only collide with a $X$ error or the $X$ part of the $Y$ error. But as the flag for these errors are at $c$ distance from $p$, to have a collision we would need two errors separated by $c+1$ which is longer than the maximum length of a correctable burst. And for an $X$ error to collide with another $X$ error or the $X$ error part of a $Y$ error, the distance between the $p's$ would have to be $2c$ which is obviously impossible.

If the sub-syndrome does not appear in the look-up table then we conclude that none of the errors in Table~\ref{dodgyerrors} have occurred and proceed as follows.
To locate the $X$ errors, note that the matrix structure implies that, when calculating the syndrome, a X error in position $p$, will raise the "flag" (i.e. a ``-'' in the sub-syndrome) in the position $p-c \pmod{m}$ and $p+c \pmod{m}$. Thus, we search for this 2-point fork structure, and if we find such a structure conclude that an $X$ error has occurred in the $p$ position. Note that the sub-syndrome has length $m-1$ and it could be that $p-c \pmod{m}$ or $p+c \pmod{m}$ is equal to $m$. But this only occurs when $p=c$ or $p=m-c$ and these are specifically the cases we corrected with the lookup table, so these cases do not occur here.

To locate the $Z$ errors, note that the matrix structure implies that, when calculating the syndrome, a $Z$ error in position $p$, will raise the "flag" (i.e. a ``-'' in the sub-syndrome) in the position $p$.

Finally a $Y$ error is deduced if we deduce an $X$ and a $Z$ error in the same position.

 The overall complexity of this algorithm is linear in $n$, since we have run through the main loop $k$ times and within the loop we have a sub-loop which is run $m$ times and we check a look up table whose size depends only on $c$. Note that since $m$ also only depends on $c$, the complexity is $O(k)=O(n)$.

\end{proof}

 The look up tables for $c=1$ and $c=2$ are in Table~\ref{cis1} and Table~\ref{cis2} respectively. Note that we fix $m=4c+1$, so that the code we are using attains the Reiger bound.

As mentioned above, the look up table will have $4^{c-1}6c$ syndromes.
\begin{table}
    \centering
    \begin{tabular}{|l|l|}
    \hline
        X, I, I, I, I & -, +, -, - \\ \hline
         Y, I, I, I, I &  +, +, -, - \\ \hline
         I, I, I, X, I &  -, -, +, - \\ \hline
         I, I, I, Y, I &  -, -, +, + \\ \hline
         I, I, I, I, Z &  -, -, -, - \\ \hline
         I, I, I, I, Y &  +, -, -, + \\ \hline
    \end{tabular}
\\
    \caption{The subsyndrome look-up table for $c=1$.} 
    \label{cis1}
\end{table}

\begin{table}
\begin{footnotesize}
    \centering
    \begin{tabular}{|l|l|}
    \hline
        I, X, I, I,  I, I, I, I, I  & -, -, -, +, -, -, -, - \\ \hline
        X, X, I, I,  I, I, I, I, I  &  -, -, +, +, -, -, -, + \\ \hline
           Z, X, I, I,  I, I, I, I, I  &  +, -, -, +, -, -, -, - \\ \hline
           Y, X, I, I, I, I, I, I, I  &  +, -, +, +, -, -, -, + \\ \hline
           I, Y, I, I, I, I, I, I, I  &  -, +, -, +, -, -, -, - \\ \hline
           X, Y, I, I, I, I, I, I, I  &  -, +, +, +, -, -, -, + \\ \hline
           Z, Y, I, I, I, I, I, I, I  &  +, +, -, +, -, -, -, - \\ \hline
           Y, Y, I, I, I, I, I, I, I  &  +, +, +, +, -, -, -, + \\ \hline
           I, X, I, I, I, I, I, I, I  &  -, -, -, +, -, -, -, - \\ \hline
           I, X, X, I, I, I, I, I, I  &  +, -, -, +, +, -, -, - \\ \hline
           I, X, Z, I, I, I, I, I, I  &  -, -, +, +, -, -, -, - \\ \hline
           I, X, Y, I, I, I, I, I, I  &  +, -, +, +, +, -, -, - \\ \hline
           I, Y, I, I, I, I, I, I, I  &  -, +, -, +, -, -, -, - \\ \hline
           I, Y, X, I, I, I, I, I, I  &  +, +, -, +, +, -, -, - \\ \hline
           I, Y, Z, I, I, I, I, I, I  &  -, +, +, +, -, -, -, - \\ \hline
           I, Y, Y, I, I, I, I, I, I  &  +, +, +, +, +, -, -, - \\ \hline
           I, I, I, I, I, I, X, I, I  &  -, -, -, -, +, -, -, - \\ \hline
           I, I, I, I, I, X, X, I, I  &  -, -, -, +, +, -, -, + \\ \hline
           I, I, I, I, I, Z, X, I, I  &  -, -, -, -, +, +, -, - \\ \hline
           I, I, I, I, I, Y, X, I, I  &  -, -, -, +, +, +, -, + \\ \hline
           I, I, I, I, I, I, Y, I, I  &  -, -, -, -, +, -, +, - \\ \hline
           I, I, I, I, I, X, Y, I, I  &  -, -, -, +, +, -, +, + \\ \hline
           I, I, I, I, I, Z, Y, I, I  &  -, -, -, -, +, +, +, - \\ \hline
           I, I, I, I, I, Y, Y, I, I  &  -, -, -, +, +, +, +, + \\ \hline
           I, I, I, I, I, I, X, I, I  &  -, -, -, -, +, -, -, - \\ \hline
           I, I, I, I, I, I, X, X, I  &  +, -, -, -, +, +, -, - \\ \hline
           I, I, I, I, I, I, X, Z, I  &  -, -, -, -, +, -, -, + \\ \hline
           I, I, I, I, I, I, X, Y, I  &  +, -, -, -, +, +, -, + \\ \hline
           I, I, I, I, I, I, Y, I, I  &  -, -, -, -, +, -, +, - \\ \hline
           I, I, I, I, I, I, Y, X, I  &  +, -, -, -, +, +, +, - \\ \hline
           I, I, I, I, I, I, Y, Z, I  &  -, -, -, -, +, -, +, + \\ \hline
           I, I, I, I, I, I, Y, Y, I  &  +, -, -, -, +, +, +, + \\ \hline
           I, I, I, I, I, I, I, I, Z  &  -, -, -, -, -, -, -, - \\ \hline
           I, I, I, I, I, I, I, X, Z  &  +, -, -, -, -, +, -, - \\ \hline
           I, I, I, I, I, I, I, Z, Z  &  -, -, -, -, -, -, -, + \\ \hline
           I, I, I, I, I, I, I, Y, Z  &  +, -, -, -, -, +, -, + \\ \hline
           I, I, I, I, I, I, I, I, Y  &  -, +, -, -, -, -, +, - \\ \hline
           I, I, I, I, I, I, I, X, Y  &  +, +, -, -, -, +, +, - \\ \hline
           I, I, I, I, I, I, I, Z, Y  &  -, +, -, -, -, -, +, + \\ \hline
           I, I, I, I, I, I, I, Y, Y  &  +, +, -, -, -, +, +, + \\ \hline
           I, I, I, I, I, I, I, I, Z  &  -, -, -, -, -, -, -, - \\ \hline
           X, I, I, I, I, I, I, I, Z  &  -, -, +, -, -, -, -, + \\ \hline
           Z, I, I, I, I, I, I, I, Z  &  +, -, -, -, -, -, -, - \\ \hline
           Y, I, I, I, I, I, I, I, Z  &  +, -, +, -, -, -, -, + \\ \hline
           I, I, I, I, I, I, I, I, Y  &  -, +, -, -, -, -, +, - \\ \hline
           X, I, I, I, I, I, I, I, Y  &  -, +, +, -, -, -, +, + \\ \hline
           Z, I, I, I, I, I, I, I, Y  &  +, +, -, -, -, -, +, - \\ \hline
           Y, I, I, I, I, I, I, I, Y  &  +, +, +, -, -, -, +, + \\ \hline
    \end{tabular}
    \caption{The subsyndrome look-up table for $c=2$.} 
    \label{cis2}
    \end{footnotesize}
\end{table}

The decoding algorithm detailed above is summarised in the flow chart, Figure~\ref{flowc}.

\newpage

\begin{figure}
\begin{center}
\includegraphics[width=0.45\textwidth]{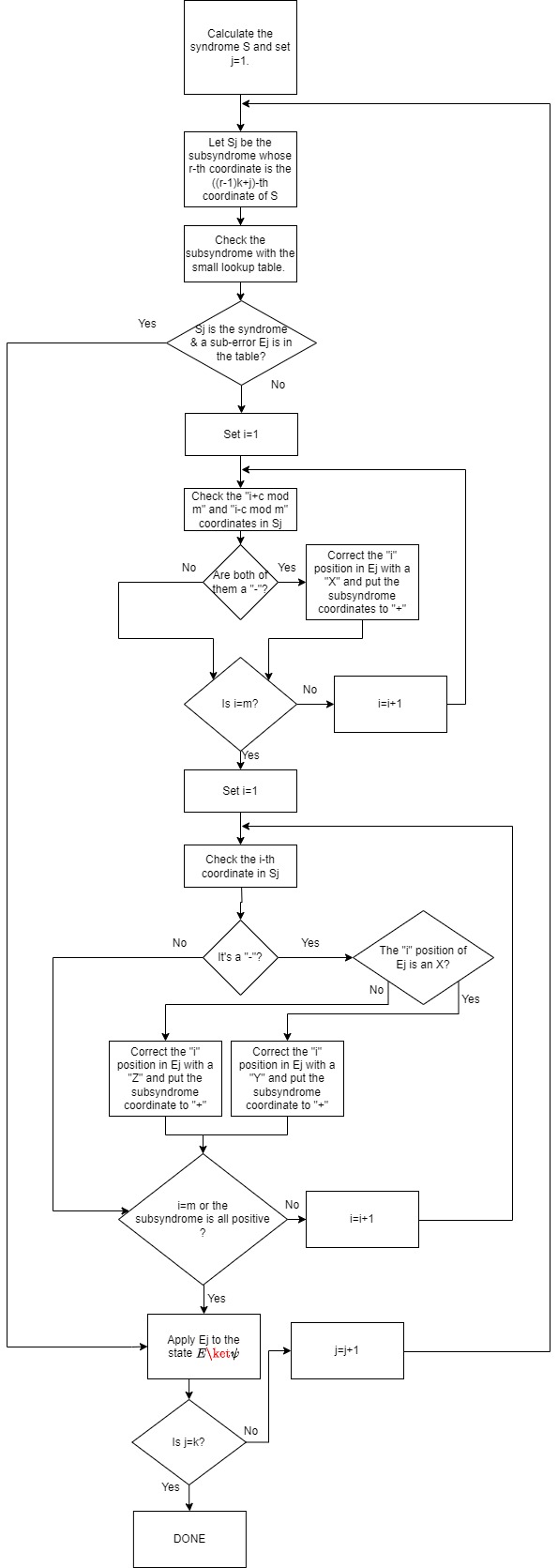}
 \caption{The linear time decoding algorithm as a flow chart.} 
 \label{flowc}
 \end{center}
\end{figure}

In conclusion, the algorithmic part scales with $n$ and the look up table scales exponentially with "c". However, since $c=l/k$ and $\ell=\frac{n-k}{4}$, this can be taken to be constant if we assume that  $k$ scales with $n$.

\begin{example}
Consider the $[[18,2]]$ code, which can correct bursts of length less than or equal to $4$. The (quasi-cyclic) matrix $G$ in this case is the following.

\begin{equation}
	\mathrm G_{[[18,2]]}=	\begin{pmatrix}
		\begin{array}{cccccccc|cccccccccccccccccc}
			
   1& 0& 0& 0& \dots 0& 0& 1& 0& 0& 0& 1& 0& 1& 0& 0& 0& 0& 0& 0& 0& 1& 0& 1& 0& 0& 0 \\
   0& 1& 0& 0& \dots 0& 0& 0& 1& 0& 0& 0& 1& 0& 1& 0& 0& 0& 0& 0& 0& 0& 1& 0& 1& 0& 0 \\
   0& 0& 1& 0& \dots 0& 0& 1& 0& 0& 0& 1& 0& 0& 0& 1& 0& 0& 0& 0& 0& 1& 0& 0& 0& 1& 0 \\
   0& 0& 0& 1& \dots 0& 0& 0& 1& 0& 0& 0& 1& 0& 0& 0& 1& 0& 0& 0& 0& 0& 1& 0& 0& 0& 1 \\
   0& 0& 0& 0& \dots 0& 0& 1& 0& 1& 0& 1& 0& 0& 0& 0& 0& 1& 0& 0& 0& 1& 0& 0& 0& 0& 0 \\
   0& 0& 0& 0& \dots 0& 0& 0& 1& 0& 1& 0& 1& 0& 0& 0& 0& 0& 1& 0& 0& 0& 1& 0& 0& 0& 0 \\
   0& 0& 0& 0& \dots 0& 0& 1& 0& 0& 0& 0& 0& 0& 0& 0& 0& 0& 0& 1& 0& 1& 0& 0& 0& 0& 0 \\
   0& 0& 0& 0& \dots 0& 0& 0& 1& 0& 0& 0& 0& 0& 0& 0& 0& 0& 0& 0& 1& 0& 1& 0& 0& 0& 0 \\
   0& 0& 0& 0& \dots 0& 0& 1& 0& 0& 0& 1& 0& 1& 0& 0& 0& 0& 0& 0& 0& 0& 0& 0& 0& 0& 0 \\
   0& 0& 0& 0& \dots 0& 0& 0& 1& 0& 0& 0& 1& 0& 1& 0& 0& 0& 0& 0& 0& 0& 0& 0& 0& 0& 0 \\
   0& 0& 0& 0& \dots 0& 0& 1& 0& 0& 0& 1& 0& 0& 0& 1& 0& 0& 0& 0& 0& 1& 0& 1& 0& 0& 0 \\
   0& 0& 0& 0& \dots 0& 0& 0& 1& 0& 0& 0& 1& 0& 0& 0& 1& 0& 0& 0& 0& 0& 1& 0& 1& 0& 0 \\
   0& 0& 0& 0& \dots 0& 0& 1& 0& 0& 0& 1& 0& 0& 0& 0& 0& 1& 0& 0& 0& 1& 0& 0& 0& 1& 0 \\
   0& 0& 0& 0& \dots 0& 0& 0& 1& 0& 0& 0& 1& 0& 0& 0& 0& 0& 1& 0& 0& 0& 1& 0& 0& 0& 1 \\
   0& 0& 0& 0& \dots 1& 0& 1& 0& 1& 0& 1& 0& 0& 0& 0& 0& 0& 0& 1& 0& 1& 0& 0& 0& 0& 0 \\
   0& 0& 0& 0& \dots 0& 1& 0& 1& 0& 1& 0& 1& 0& 0& 0& 0& 0& 0& 0& 1& 0& 1& 0& 0& 0& 0

		\end{array}
	\end{pmatrix}
\end{equation}

A set of generators for the stabilizer group $S$ is given by the row of the matrix

\begin{equation}
	\begin{pmatrix}
		\begin{array}{cccccccccccccccccc}
  X& I& Z& I& Z& I& I& I& I& I& I& I& Z& I& Z& I& X& I \\  I& X& I& Z& I& Z& I& I& I& I& I& I& I& Z& I& Z& I& X \\
   I& I& Y& I& I& I& Z& I& I& I& I& I& Z& I& I& I& Y& I \\  I& I& I& Y& I& I& I& Z& I& I& I& I& I& Z& I& I& I& Y \\
   Z& I& Z& I& X& I& I& I& Z& I& I& I& Z& I& I& I& X& I \\  I& Z& I& Z& I& X& I& I& I& Z& I& I& I& Z& I& I& I& X \\
   I& I& I& I& I& I& X& I& I& I& Z& I& Z& I& I& I& X& I \\  I& I& I& I& I& I& I& X& I& I& I& Z& I& Z& I& I& I& X \\
   I& I& Z& I& Z& I& I& I& X& I& I& I& I& I& I& I& X& I \\  I& I& I& Z& I& Z& I& I& I& X& I& I& I& I& I& I& I& X \\
   I& I& Z& I& I& I& Z& I& I& I& X& I& Z& I& Z& I& X& I \\  I& I& I& Z& I& I& I& Z& I& I& I& X& I& Z& I& Z& I& X \\
   I& I& Z& I& I& I& I& I& Z& I& I& I& Y& I& I& I& Y& I \\  I& I& I& Z& I& I& I& I& I& Z& I& I& I& Y& I& I& I& Y \\
   Z& I& Z& I& I& I& I& I& I& I& Z& I& Z& I& X& I& X& I \\  I& Z& I& Z& I& I& I& I& I& I& I& Z& I& Z& I& X& I& X

\end{array}
	\end{pmatrix}.
\end{equation}

The base code of this interleaved code is a $[[9,1]]$ code, capable of correcting burst errors of length at most two. The generator matrix of the base code is

\begin{equation} \label{eqn11}
	\mathrm G_{[[9,1]]}=	\begin{pmatrix}
		\begin{array}{ccccccccc|ccccccccc}

			1 & 0 & 0 & 0 & 0 & 0 & 0 & 0 & 1 & 0 & 1 & 1 & 0 & 0 & 0 & 1 & 1 & 0\\
			0 & 1 & 0 & 0 & 0 & 0 & 0 & 0 & 1 & 0 & 1 & 0 & 1 & 0 & 0 & 1 & 0 & 1\\
			0 & 0 & 1 & 0 & 0 & 0 & 0 & 0 & 1 & 1 & 1 & 0 & 0 & 1 & 0 & 1 & 0 & 0\\
            0 & 0 & 0 & 1 & 0 & 0 & 0 & 0 & 1 & 0 & 0 & 0 & 0 & 0 & 1 & 1 & 0 & 0\\
            0 & 0 & 0 & 0 & 1 & 0 & 0 & 0 & 1 & 0 & 1 & 1 & 0 & 0 & 0 & 0 & 0 & 0\\
			0 & 0 & 0 & 0 & 0 & 1 & 0 & 0 & 1 & 0 & 1 & 0 & 1 & 0 & 0 & 1 & 1 & 0\\
			0 & 0 & 0 & 0 & 0 & 0 & 1 & 0 & 1 & 0 & 1 & 0 & 0 & 1 & 0 & 1 & 0 & 1\\
            0 & 0 & 0 & 0 & 0 & 0 & 0 & 1 & 1 & 1 & 1 & 0 & 0 & 0 & 1 & 1 & 0 & 0

		\end{array}
	\end{pmatrix}.
\end{equation}

A set of generators for the stabilizer group given by this matrix is given by the rows of the matrix

\begin{equation}
	\begin{pmatrix}
		\begin{array}{ccccccccc}

   X& Z& Z& I& I& I& Z& Z& X \\ 
   I& Y& I& Z& I& I& Z& I& Y \\
   Z& Z& X& I& Z& I& Z& I& X \\
   I& I& I& X& I& Z& Z& I& X \\
   I& Z& Z& I& X& I& I& I& X \\
   I& Z& I& Z& I& X& Z& Z& X \\
   I& Z& I& I& Z& I& Y& I& Y \\
   Z& Z& I& I& I& Z& Z& X& X \\

		\end{array}
	\end{pmatrix}.
\end{equation}

Observe that by deleting every even column ($k=2$) from the first matrix, one obtains the second matrix.

This allows us to partition the syndromes into sub-syndromes. 

Suppose we measure the syndrome

    $$\mathrm {
  s=( + + - - - + - + + + - - - + + +).
    }
    $$

We can partition this into $k=2$ sub-syndromes:

 $$s_1=( 
    \mathrm {+ - - - + - - +}), \ \ s_2=(\mathrm{ + - + + + - + +
    })
    $$

Let us assume that we have encoded the information in the state $\ket{\psi}$ and that some Pauli error $E$ has occurred and we have calculated the syndrome on the state $E\ket{\psi}$. Assuming the maximum burst length of errors in $E$ has length $\ell=4$ we will demonstrate how to correct $E$ and recover $\ket{\psi}$.

Consider the first subsyndrome:

$$s_1=( 
    \mathrm {+ - - - + - - +})
    $$
    
We first check if we have the sub-syndrome in the lookup table, Table~\ref{cis2}, and the answer is no. Then starting from $p=1$ we check the positions $p-2 \ mod \ 9$ and $p+2 \ mod \ 9$. We see that the first position where both positions are a $-$ is when $p=4$, with position number 2 and 6 being $-$. So we will start building our error with an $X$ in the $4th$ position:

 $$Error:\mathrm {I I I X I IIII}$$

 We change the subsyndrome accordingly:

 $$ s_1=( 
    \mathrm {+ + - - + + - +})
    $$
We continue and observe that the same occurs for position $p=5$. Thus, we update the error accordingly.

 $$Error:\mathrm {I I I X X IIII}$$

 And change the subsyndrome accordingly:

 $$s_1=(  
    \mathrm {+ + + - + + + +})$$

We continue to check the positions $p-2 \ mod \ 9$ and $p+2 \ mod \ 9$ and conclude that there are no further values of $p$ for which both positions are a $-$.

We then seek to identify the $Z$ errors. There is just the $p=4$ position which is a $-$. Since we already have a $X$ error in position $4$, we change that into a $Y$. Hence, the error is

$$E_1=\mathrm {I I I Y X IIII}.$$

We repeat the same part of the algorithm with the second sub-syndrome

$$s_2=( 
    \mathrm { + - + + + - + +
    })
    $$
and obtain

$$E_2=\mathrm {IIIXIIIII  }.$$

Interleaving the errors $E_1$ and $E_2$ we obtain the full error
$$E=\mathrm {   I I I I I I Y X X I I I I I I I I I}.
$$

\end{example}

\section{Simulations}

We performed some simulations to compare our codes to standard codes and a code from \cite{FHCCL2018}. We followed \cite{FHCCL2018} and used a Markovian correlated depolarizing quantum channel represented in Figure \ref{fig:markov}. The $0$ state represents a non-error occurrence in the channel and the $1$ state represents an error. For a length $n$ word, $p$ is probability the first qubit will have an error and the Markovian state will be 1 and $1-p$ is the probability the first qubit doesn't have an error and the Markovian state will be 0. For the next $n-1$ qubits, apply $n-1$ steps in the Markovian channel where $p_{00}$ is the probability of getting no error after the last qubit had no error, $p_{11}$ is the probability of getting an error after the last qubit was an error and $p_{10}$ and $p_{01}$ are the probabilities of getting an error qubit after a no-error qubit and viceversa.

Suppose $p_{00} = (1-\mu)(1-p)+\mu$, $p_{11} = (1-\mu)p+\mu$, $p_{01} = (1-\mu)(1-p)$ and $p_{10} = p(1-\mu)$ where $p$ is the probability of error and each error ($X,Y,Z$) has the same probability and $\mu$ is the correlation degree from $0$ to $1$.

\begin{figure}
    \centering
    \caption{Markovian correlated depolarizing quantum channel}
    \label{fig:markov}
    
    \begin{center}
	\begin{tikzpicture}[->, >=stealth', auto, semithick, node distance=3cm]
	\tikzstyle{every state}=[fill=white,draw=black,thick,text=black,scale=1]
	\node[state]    (A)                     {$0$};
	\node[state]    (B)[right of=A]   {$1$};
;
	\path
	(A) edge[loop left]			node{$p_{00}$}	(A)
	(B) edge[bend left,below]	node{$p_{01}$}	(A)
	(A) edge[bend left,above]	node{$p_{10}$}	(B)
	(B) edge[loop right]			node{$p_{11}$}	(B);
	\end{tikzpicture}
\end{center}
    
\end{figure}
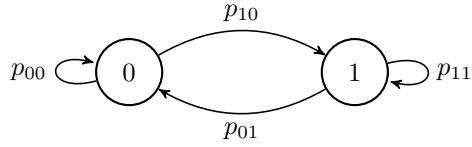	

We used the $[[35,7]]$ code obtained from our construction, that can correct burst errors of length $7$ and random errors of weight $3$ and compared it to the $[[35,7]]$  code constructed in \cite{FHCCL2018} that can correct bursts of length $6$ and $3$ random errors.

\begin{figure}
    \centering
    \includegraphics[width=0.75\textwidth]{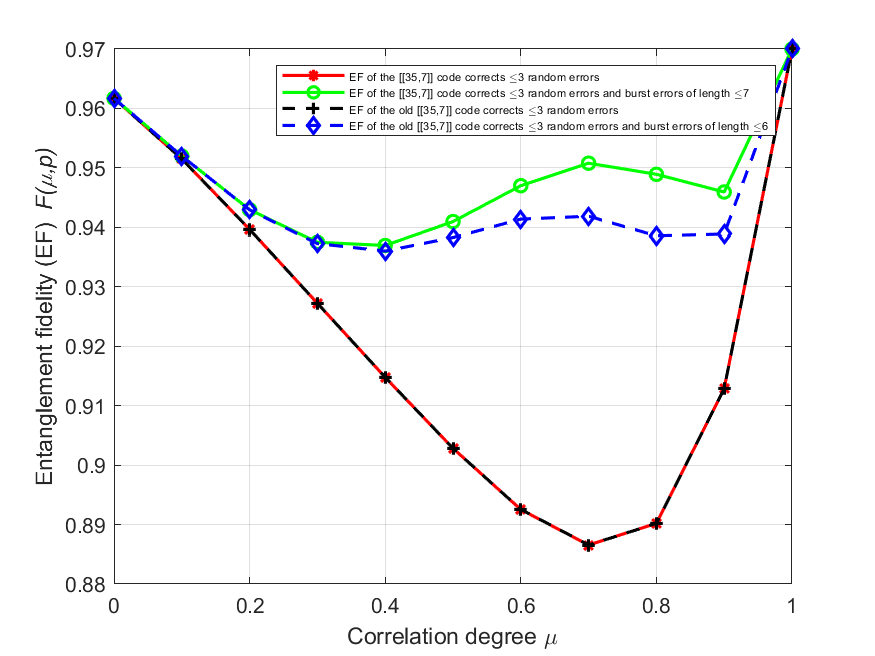}
    \caption{Entanglement fidelity of the 2 codes with respect to the correlation degree $0 \leq \mu \leq 1$. The error probability is $p=3 \times 10^{-2}$.}
    \label{fig:ef1}
\end{figure}

\begin{figure}
    \centering
    \includegraphics[width=0.75\textwidth]{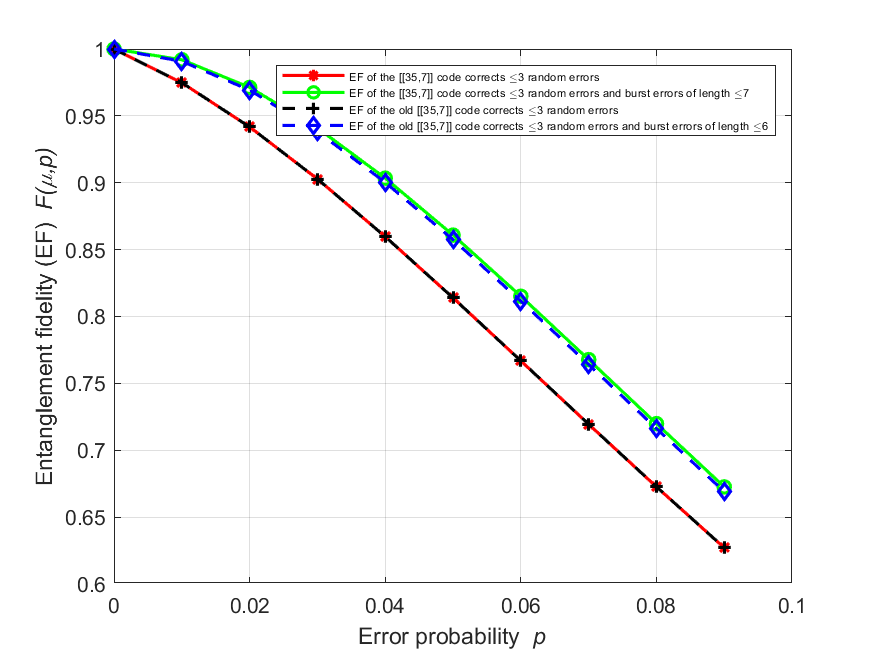}
    \caption{Entanglement fidelity of the 2 codes with respect to the error probability $p=10^{-3} \leq p \leq 10^{-1}$, where $\mu=0.5$.}
    \label{fig:ef2}
\end{figure}
\begin{figure}
    \centering
    \includegraphics[width=0.75\textwidth]{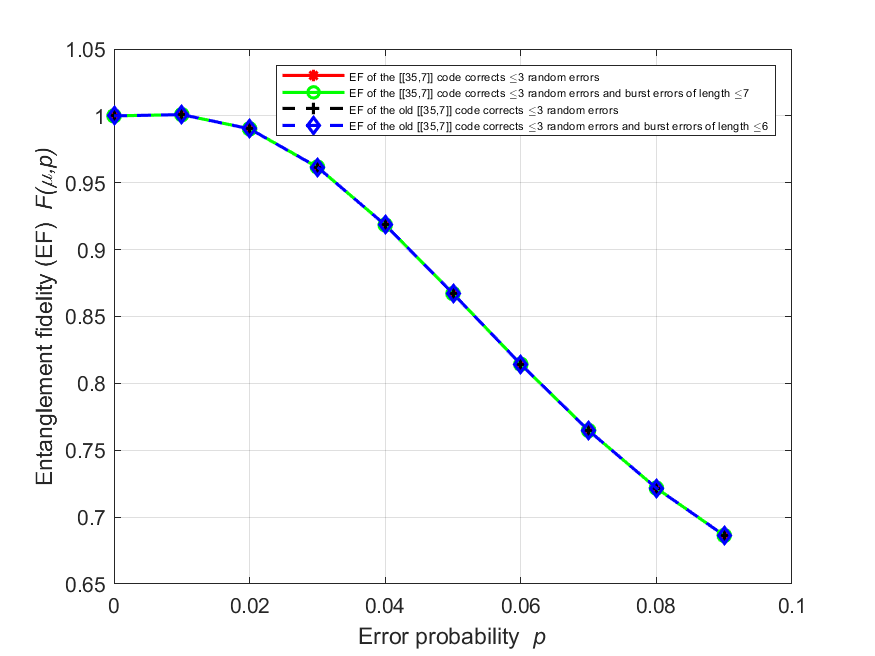}
    \caption{Entanglement fidelity of the 2 codes with respect to the error probability $10^{-3} \leq p \leq 10^{-1}$ where $\mu=0$.}
    \label{fig:ef3}
\end{figure}

In Figure \ref{fig:ef1} is clear that the entanglement fidelity is improved when using the quantum burst error-correcting codes.

With Figure \ref{fig:ef2} and \ref{fig:ef3} we can see that when the correlation is $0$, the code performs the same as a standard code, but when the correlation is $0.5$, there is a constant distance between the entanglement fidelity curves.

\section{Conclusion}
In conclusion, in this article we have introduced a method to construct quantum CRC codes from classical CRC codes. We note that this method can also be used with any classical code to construct a quantum stabilizer code. The classical CRC codes which obtain the Reiger bound, will give quantum CRC codes which attain the quantum Reiger bound. Since there are classical CRC codes which obtain the Reiger bound for all length $n$ and dimension $k$, we conclude that the quantum Reiger bound is always attainable.

We also presented a fast decoding algorithm for a family of quantum CRC codes with parameters $[[mk, k]]$, where $m \geq 5$.

\begin{table} [h]
\begin{footnotesize}
\begin{center}
\begin{tabular}{|c|c|l|}
\hline
$n$ & $k$ & $g(X)$ with the c-property \\ \hline
7   &    4   &    $X^3+X+1$ \\
7   &    4   &    $X^3+X^2+1$ \\
7   &    3   &    $X^4+X^3+X^2+1$ \\
7   &    3   &    $X^4+X^2+X+1$ \\
9   &    3   &    $X^6+X^3+1$ \\
9   &    2   &    $X^7+X^6+X^4+X^3+X+1$ \\
15   &    10   &    $X^5+X^4+X^2+1$ \\
15   &    10   &    $X^5+X^3+X+1$ \\
15   &    9   &    $X^6+X^5+X^4+X^3+1$ \\
15   &    9   &    $X^6+X^3+X^2+X+1$ \\
15   &    7   &    $X^8+X^7+X^6+X^4+1$ \\
15   &    7   &    $X^8+X^4+X^2+X+1$ \\
15   &    8   &    $X^7+X^3+X+1$ \\
15   &    8   &    $X^7+X^6+X^4+1$ \\
15   &    6   &    $X^9+X^6+X^5+X^4+X+1$ \\
15   &    6   &    $X^9+X^8+X^5+X^4+X^3+1$ \\
15   &    5   &    $X^{10}+X^5+1$ \\
15   &    5   &    $X^{10}+X^8+X^5+X^4+X^2+X+1$ \\
15   &    5   &    $X^{10}+X^9+X^8+X^6+X^5+X^2+1$ \\
15   &    3   &    $X^{12}+X^9+X^6+X^3+1$ \\
15   &    4   &    $X^{11}+X^{10}+X^6+X^5+X+1$ \\
15   &    4   &    $X^{11}+X^{10}+X^9+X^8+X^6+X^4+X^3+1$ \\
15   &    4   &    $X^{11}+X^8+X^7+X^5+X^3+X^2+X+1$ \\
15   &    2   &    $X^{13}+X^{12}+X^{10}+X^9+X^7+X^6+X^4+X^3+X+1$ \\
21   &    14   &    $X^7+X^6+X^5+X^4+X^3+1$ \\
21   &    14   &    $X^7+X^4+X^3+X^2+X+1$ \\
21   &    12   &    $X^9+X^8+X^5+X^4+X^2+X+1$ \\
21   &    12   &    $X^9+X^8+X^7+X^5+X^4+X+1$ \\
21   &    12   &    $X^9+X^7+X^6+X^5+X^3+X^2+X+1$ \\
21   &    12   &    $X^9+X^8+X^7+X^6+X^4+X^3+X^2+1$ \\
21   &    10   &    $X^{11}+X^{10}+X^9+X^8+X^7+X^6+X^2+X+1$ \\

\hline
\end{tabular}
\caption{List of $g(X)\not\in \{X+1,(X^n+1)/(X+1)\}$ which have the c-property.} \label{cproptable1}
\end{center}
\end{footnotesize}
\end{table}

\begin{table} [h]
\begin{footnotesize}
\begin{center}
\begin{tabular}{|c|c|l|}
\hline
21   &    10   &    $X^{11}+X^{10}+X^9+X^5+X^4+X^3+X^2+X+1$ \\
21   &    7   &    $X^{14}+X^7+1$ \\
21   &    9   &    $X^{12}+X^{11}+X^8+X^6+X^3+X^2+1$ \\
21   &    9   &    $X^{12}+X^{10}+X^9+X^6+X^4+X+1$ \\
21   &    6   &    $X^{15}+X^{14}+X^{13}+X^{12}+X^{10}+X^9+X^8+X^5+X^4+X^2+1$ \\
21   &    6   &    $X^{15}+X^{13}+X^{11}+X^{10}+X^7+X^6+X^5+X^3+X^2+X+1$ \\
21   &    9   &    $X^{12}+X^6+X^3+1$ \\
21   &    9   &    $X^{12}+X^9+X^6+1$ \\
21   &    6   &    $X^{15}+X^{14}+X^8+X^7+X+1$ \\
21   &    8   &    $X^{13}+X^{11}+X^9+X^8+X^7+X^6+X^4+X^2+X+1$ \\
21   &    8   &    $X^{13}+X^{12}+X^{11}+X^9+X^7+X^6+X^5+X^4+X^2+1$ \\
21   &    5   &    $X^{16}+X^{12}+X^{11}+X^8+X^6+X^4+X^3+X^2+X+1$ \\
21   &    5   &    $X^{16}+X^{15}+X^{14}+X^{13}+X^{12}+X^{10}+X^8+X^5+X^4+1$ \\
21   &    7   &    $X^{14}+X^{11}+X^{10}+X^9+X^7+X^6+X^5+X+1$ \\
21   &    7   &    $X^{14}+X^{13}+X^9+X^8+X^7+X^5+X^4+X^3+1$ \\
21   &    4   &    $X^{17}+X^{15}+X^{14}+X^{10}+X^8+X^7+X^3+X+1$ \\
21   &    4   &    $X^{17}+X^{16}+X^{14}+X^{10}+X^9+X^7+X^3+X^2+1$ \\
21   &    3   &    $X^{18}+X^{15}+X^{12}+X^9+X^6+X^3+1$ \\
21   &    6   &    $X^{15}+X^{14}+X^{12}+X^9+X^8+X^5+X^2+1$ \\
21   &    6   &    $X^{15}+X^{13}+X^{10}+X^7+X^6+X^3+X+1$ \\
21   &    3   &    $X^{18}+X^{17}+X^{16}+X^{14}+X^{11}+X^{10}+X^9+X^7+X^4+X^3+X^2+1$ \\
21   &    3   &    $X^{18}+X^{16}+X^{15}+X^{14}+X^{11}+X^9+X^8+X^7+X^4+X^2+X+1$ \\
21   &    2   &    $X^{19}+X^{18}+X^{16}+X^{15}+X^{13}+X^{12}+X^{10}+X^9+X^7+X^6+X^4+X^3+X+1$ \\ 
23   &    12   &    $X^{11}+X^9+X^7+X^6+X^5+X+1$ \\
23   &    12   &    $X^{11}+X^{10}+X^6+X^5+X^4+X^2+1$ \\
25   &    5   &    $X^{20}+X^{15}+X^{10}+X^5+1$ \\
25   &    4   &    $X^{21}+X^{20}+X^{16}+X^{15}+X^{11}+X^{10}+X^6+X^5+X+1$ \\
27   &    9   &    $X^{18}+X^9+1$ \\
27   &    8   &    $X^{19}+X^{18}+X^{10}+X^9+X+1$ \\
27   &    3   &    $X^{24}+X^{21}+X^{18}+X^{15}+X^{12}+X^9+X^6+X^3+1$ \\
27   &    2   &    $X^{25}+X^{24}+X^{22}+X^{21}+X^{19}+X^{18}+X^{16}+X^{15}+X^{13}+X^{12}+X^{10}+X^9+X^7+X^6+X^4+X^3+X+1$ \\

 \hline
\end{tabular}
\caption{List of $g(X)\not\in \{X+1,(X^n+1)/(X+1)\}$ which have the c-property.} 
\label{cproptable2}
\end{center}
\end{footnotesize}
\end{table}


\begin{thebibliography}{1}






\bibitem{AMD2021} 
W. An, M. M\'edard and K. R. Duffy, CRC Codes as Error Correction Codes, ICC 2021 - IEEE International Conference on Communications, Montreal, QC, Canada, 2021, pp. 1-6.


\bibitem{BCH2021} S. Ball, F. Huber and A. Centelles, {Quantum error-correcting codes and their geometries}, {\it Ann. Inst. Henri Poincar\'e Comb. Phys. Interact}., to appear. \url{arXiv:2007.05992}.



\bibitem{BrunLidar2013}
T. A. Brun and D. E. Lidar, 
{\it Quantum Error Correction},
Cambridge University Press, 2013.


\bibitem{CRSS1998} A. R. Calderbank, E. M. Rains, P. W. Shor, and N. Sloane, Quantum
error correction via codes over $GF(4)$, {\it IEEE Trans. Inf. Theory}, {\bf 44} (1998) 1369--1387.


\bibitem{CGLM2014} F. Caruso, V. Giovannetti, C. Lupo, and S. Mancini, Quantum channels and memory effects, {\it Rev. Mod. Phys.}, {\bf 86} (2014), 1203--.


\bibitem{FHCCL2018} J. Fan, M Hseih, H. Chen, H. Chen and Y. Li, Construction and Performance of Quantum Burst Error Correction Codes for Correlated Errors, \url{arxiv:1801.03961}.

 \bibitem{GGMG}
D. G. Glynn, T. A. Gulliver, J. G. Maks and M. K. Gupta,
{\it The Geometry of Additive Quantum Codes}, unpublished manuscript. (available online at 
\url{https://www.academia.edu/17980449/})


\bibitem{Gottesman2009} 
D. Gottesman, An Introduction to Quantum Error Correction and Fault-Tolerant Quantum Computation, in {\it Quantum Information Science and Its Contributions to Mathematics, Proceedings of Symposia in Applied Mathematics} {\bf 68}, pp. 13-58 (Amer. Math. Soc., Providence, Rhode Island, 2010). (available online at \url{https://arxiv.org/abs/0904.2557}).

\bibitem{GottesmanThesis}
D. Gottesman, 
{\it Stabilizer Codes and Quantum Error Correction},
PhD Thesis (1997)
(available online at \url{https://arxiv.org/abs/quant-ph/9705052}).




\bibitem{codetables}
M.~Grassl, 
Bounds on the minimum distance of linear codes and quantum codes,
(available online at \url{http://www.codetables.de}).




\bibitem{KKKS2006} A. Ketkar, A. Klappenecker, S. Kumar and P. K. Sarvepalli, Nonbinary stabilizer codes over finite fields, {\em IEEE Trans. Inform. Theory}, {\bf 52} (2006) 4892--4914. (available online at \url{https://arxiv.org/abs/quant-ph/0508070})


\bibitem{KL1997} E. Knill and R. Laflamme. Theory of quantum error-correcting codes, {\it Phys. Rev. A}, {\bf 55} (1997) 900--911.

\bibitem{CRC2022} P. Koopman, {\it 32-Bit Cyclic Redundancy Codes for Internet Applications}, {\it Proceedings International Conference on Dependable Systems and Networks, Washington, DC, USA}  pp. 459-468, (2002).


\bibitem{NC2000}
M. Nielsen and I. Chuang,
{\it Quantum Computation and Quantum Information},
Cambridge University Press, 2000.


\bibitem{CRC}
W. W. Peterson D. T. Brown, Cyclic Codes for Error Detection, {\it Proceedings of the IRE}, {\bf 49} (1961) 228--235.

\bibitem{Reiger1960} S. Reiger, Codes for the correction of 'clustered' errors,  {\it IRE Transactions on Information Theory}, {\bf 6} (1960) 16--21.

\bibitem{Shor1995} P. W. Shor, Scheme for reducing decoherence in quantum memory, {\it Phys. Rev. Lett.}, {\bf 77} (1996) 793--797.

\bibitem{Singleton1964} R. C. Singleton, Maximum distance q-nary codes, {\it IEEE Trans. Inf. Theory}, {\bf 10} (1964) 116--118.

\bibitem{Nature2021} C. D. Wilen, S. Abdullah, N. A. Kurinsky, C. Stanford, L. Cardani, G. D'Imperio,
C. Tomei, L. Faoro, L. B. Ioffe, C. H. Liu, A. Opremcak, B. G. Christensen, J. L. DuBois and R. McDermott, Correlated charge noise and relaxation errors in superconducting qubits, {\it Nature}, {\bf 594}, 369--373.

\bibitem{YCO2007} S. Yu, Q. Chen  and C. H. Oh, Graphical Quantum Error-Correcting Codes, (2007) {\tt arXiv:0709.1780}.




\end{thebibliography}
\end{document}